\documentclass{amsart}
\usepackage{amsmath}
\usepackage{amssymb}
\usepackage{caption}
\usepackage{pdfpages}
\usepackage{url}
\usepackage{epigraph}
\usepackage[T2A]{fontenc}
\usepackage[utf8x]{inputenc}
\usepackage[english]{babel}

\newcommand\RR{\mathbb R}

\newtheorem{theorem}{Theorem}

\newtheorem{lemma}{Lemma}
\newtheorem{corollary}{Corollary}

\title[New mechanism for repeated posted price auction without discounting]{New mechanism for repeated posted price auction without discounting\\}

\author[N. Kalinin]{Nikita Kalinin}\thanks{National Research University Higher School of Economics, Soyuza Pechatnikov str., 16, St. Petersburg, Russian Federation. Support from the Basic Research Program of the National Research University Higher School of Economics is gratefully acknowledged. Supported in part by Young Russian Mathematics award.}  
\email{nkalinin@hse.ru}
\address{National Research University Higher School of Economics, Soyuza Pechatnikov str., 16, St. Petersburg, Russian Federation} 

\begin{document}
\begin{abstract} On ad exchange platforms the place for advertisement is sold through different kinds of auctions. However, it is not uncommon the situation where the seller repeatedly encounters only one buyer, thus the posted price auction degenerates into a monopoly-monopsony game with asymmetric information and nearly an infinite number of rounds; on each round the seller proposes the price and the buyer accepts or rejects it. 

I learned this problem from a discussion with members of Yandex research team and my main motivation was to find an incentive-compatible seller's strategy. In this short paper such a strategy is proposed and a corresponding distortion at the top type lower bound (Spence-Mirrlees property, actually) for the surplus of the buyer is established; this shows that the proposed strategy is the best possible.

The key ingredients are the following. The main leash that the buyer has is the frequency of accepted deals. Once this frequency (as a function on the buyer's type) is fixed, the strategy randomly chooses between the {\it rewarding} price which incentivises the buyer to reveal his type (the higher the type, the more average surplus the buyer has), the {\it adaptation} price which allows the buyer to communicate that his type is higher then the current guess of the cook, and the {\it type confirmation} price which disincentivises the buyer to pretend that his type is higher than it is.
%The arguments are simple and consist of looking at the problem from different angles.
\end{abstract}

%
%\begin{abstract}
%A repeated posted price auction consists of a cook and a fisher. Every day the fisher proposes a unit of good (think of a fish or a place for advertising) for a certain price and the cook accepts or rejects the deal. In case of acceptance, the surplus of the cook is $q^*-Q$  and the revenue of the fisher is $Q$ where $Q$ is the today price and $q^*$ is a private cook's valuation of the good. Otherwise the surplus and the revenue are both $0$.
%
%
%A naive cook would accept if $q^*\geq Q$ but then the fisher quickly determines $q^*$ and the cook's surplus $q^*- Q$ will become zero. Therefore the cook tries to trick the fisher and is afraid to buy even for small $Q$, which may crash the market. As for the fisher, any price is better than nothing (the fish rots and the place for the advertising is empty) but setting low price is not optimal. 
%%this type is bargaining if tough for both.
%
%I prove certain impossibility results and propose a novel mechanism, encouraging the cook to truthfully reveal $q^*$.
%\end{abstract}

\maketitle

%\rightline{Where the bourgeois economists saw a relation between things}
%\rightline{ (the exchange of one commodity for another)}
%\rightline{ Marx revealed a relation between people.} 
%\rightline{V.I. Lenin, {\it The Three Sources and Three Component Parts of Marxism (1913)}}

%I certainly don't want to suggest that the universe ''is'' this picture or anything like that. but it is not unlikeky that some of the features 
\rightline{I foresee the progress of game theory as depending on successive}
\rightline{reductions in the base of common knowledge required}
\rightline{to conduct useful analyses  of practical problems.}
\rightline{Only by repeated weakening of common knowledge assumptions }
\rightline{will the theory approximate reality.}
\rightline{R. Wilson, {\it Game-Theoretic Analyses of Trading Processes}.}

\section{Setup}
\subsection{The game}

The {\it repeated posted price auction}, also known as the fishmonger problem, is an archetypical repeated monopoly-monopsony game with asymmetric information. This is a game with two players: the seller ({\it fisher}) and the buyer ({\it cook}). 
In each round $i=1,2,\dots$ the fisher proposes to the cook a unit of the good (think of fish or a place for advertisement) with the price $Q_i$. Let the cook's valuation of this good be $q\in \RR_{>0}$, the {\it type} of the cook. In each round the cook has two options: accept the deal (and then we set $a_i=1$) or reject it (we set $a_i=0$). In this round, the cook's {\it surplus} is $a_i(q-Q_i)$ and the fisher's {\it revenue} is $a_iQ_i$. This is an example of {\it one-sided ignorance}: the seller does not know the value of his good for the buyer. 

The repeated posted price auction appears in ad exchange conducted by the leading Internet companies, see \cite{amin2013learning,devanur2015perfect, drutsa2017consistency,drutsa2017horizon,mohri2015revenue,mohri2014optimal,kleinberg2003value} and the references therein. It is common to assume that $q$ is drawn from a known distribution. In line with \cite{bergemann2005robust}, recently there have appeared a bunch of papers (e.g. \cite{bergemann2008pricing,babaioff2010detail,goldner2016prior,tang2018price,babaioff2011posting}) which weaken this assumption since it is practically unrealistic if the buyer is unique of their kind. We will study the case when the fisher has no information at all about $q$ (it is even worse than the worst case in terminology of \cite{kleinberg2003value} where $q$ was at least in the interval $[0,1]$).  Consequently, to make sense of the problem we must further assume that  the game has an infinite number of rounds. 

It is also common to discount the surplus of the cook or introduce the stopping time, thus making the problem more accessible: the cook prefers less loss now to the bigger gain in the future. However, it is hard to justify any choice of discount factor in practice because nowadays many auctions are run by robots and  are performed many times per second. So we unleash the cook by imposing no discounting on his surplus.

\subsection{Contribution of this paper} A {\it distortion at the top} upper bound for the surplus of the fisher is established. Namely, if the fisher commits to a certain strategy in advance, then the higher the type $q$ of the cook, the bigger share of the total welfare $q$ distributed in each round goes to the cook. 

Then, we propose a strategy such that if the fisher {\it commits} to it, then the optimal in expectation response of the cook will be to play {\it naively}, i.e. accept the deal if and only if $Q_i\leq q$, thus this strategy is {\it incentive-compatible}. Furthermore, the revenue of the fisher in this strategy attains the upper bound discussed above, so this bound is tight. The proofs are technically simple and consist merely of changing the point of view.

The proposed mechanism is {\it credible} (cf. \cite{akbarpour2018credible}), the fisher exercises ``the power to commit'': the cook can verify that the fisher is using exactly this strategy, so there will be no mistrust.

\subsection{The objectives of the players}
Fix a strategy $S_{fisher}$ of the fisher and a strategy $S_{cook}$ of the cook.
In the spirit of Wald's maxmin model, we denote $$Fisher(S_{fisher},S_{cook},q)=\liminf\limits_{k\to\infty}\frac{1}{k}\sum_{i=1}^ka_iQ_i,$$ 
$$Cook(S_{fisher},S_{cook},q)=\liminf\limits_{k\to\infty}\frac{1}{k}\sum_{i=1}^ka_i(q-Q_i),$$
and say that the fisher's objective is to maximize $Fisher(S_{fisher},S_{cook},q)$ and the cook's objective is to maximize $Cook(S_{fisher},S_{cook},q)$. Roughly speaking, they maximize the minimal average revenue and the minimal average surplus respectively.   Since the total welfare is $q=(q-Q_i)+Q_i$ one may think than in each round the fisher and the cook divide $q$ among themselves, if the deal is accepted. 

A {\it strategy} (we allow mixed strategies too) is any rule which maps the previous history of the game to a distribution from where the next price (or a decision to accept or reject) is drawn.

%I propose a strategy $S$ for the fisher and argue why the resulting revenue is asymptotically the best that the fisher can achieve using strategies which are known in advance to the cook.
%
%A speculative anecdotal corollary that we may derive from the mechanism proposed is that in order to incentivise an employee with unknown (and hidden) potential to reveal it, an employer might randomly provide this employee some free benefits with probability growing as the importance of the work being done increases\footnote{This is known for quite a while: given several types of players on the market, you must provide the players of higher types with more benefits, cf. \cite{myerson1981optimal}. In other words, it is possible to exploit the poorest, make some use of the middle class, and have only small (relative to their budgets) benefits from the rich.}.

\section{Observables and Spence-Mirrlees property}
Fix any strategy $S_{fisher}$ of the fisher. Suppose that the cook is of the type $q'$ and the cook knows $S_{fisher}$ in advance and chooses a strategy $S_{cook}^0$. Let them play and write the history of all the moves, i.e. $a_i,Q_i,i=1,\dots$. Taking the physical point of view we may try to extract {\it observables}, i.e. some quantitative data from this history. The first observable coming in mind is the empirical {\it proportion of accepted deals}. Namely, define

\begin{equation}
\label{eq_prob}
p(q')=p(q',S_{fisher},S_{cook}^0)=\liminf_{n\to\infty}\frac{1}{n}\sum_{i=1}^n a_i.
\end{equation}
Again, we would like to take the limit of the average frequency of accepted deals. Since it may not exist, we consider $\liminf$ instead.

\begin{lemma}
\label{lemma_grows}
Suppose that $p(q')>0$. Then, for each $\varepsilon>0$, for $q$ big enough we have 
$$\min Fisher(S_{fisher},S,q)\leq q\cdot(\varepsilon+1-p(q')),$$
where we take the minimum by all strategies $S$ of the cook. In other words, the fisher can not guarantee himself more than the share $1-p(q')$ of the total welfare for all sufficiently large $q\in\RR_{>0}$.
\end{lemma}
\begin{proof}

For $q>q'$ by playing $S_{cook}^0$ the cook of type $q$ has  $$Cook(S_{fisher},S_{cook}^0,q)\geq Cook(S_{fisher},S_{cook}^0,q')+(q-q')p(q').$$ 
This gives
$$\min Fisher(S_{fisher},S,q)\leq Fisher(S_{fisher},S_{cook}^0,q)\leq q - Cook(S_{fisher},S_{cook}^0,q')-(q-q')p(q') = $$

$$=q(1-p(q')) + \big[q'p(q')-Cook(S_{fisher},S_{cook}^0,q')\big]<q\cdot(1-p(q')+\varepsilon)$$
for $q$ big enough, which finishes the proof.
\end{proof}

This proof amounts to the fact that by pretending of being of the type $q'$ the cook takes roughly at least $1-p(q')$ share of the total welfare in each round.

%\section{Mechanism for repeated posted price auction}

%
%Suppose that for some strategy of the cook
% I show a negative estimates: the share of the fisher's revenue  in the total welfare tends to zero in all reasonable strategies. Then I propose a mechanism with fisher's commitment???????? (fisher does'nt know the strategy of the bueyr) which can not be improved in average with respect to the above estimates.
%
%
%
%
%we assume that both limits exist. We also assume that the proportion of accepted deals,  is also well defined for each game that we consider.
%
%\begin{remark} One can imagine a game where, initially, in a lot of rounds  the cook refuses regardless the price, then the fisher proposes small $Q_i=Q$ for a bigger number of rounds and the cook accepts, then for even bigger number of times the cook again refuses, etc. That is not to say that the players act inanely but in such a game all three limits do not exist.
%\end{remark}

Let the cook of type $q$ play a strategy $S_{cook}(q)$ against $S_{fisher}$. It is natural to assume that the revenue of the cook of type $q$ using strategy $S_{cook}(q)$ is at least that of playing the strategy $S_{cook}(q')$ with $q'<q$, because otherwise we may set $S(q):=S(q')$.

Using the history of playing $S_{cook}(q)$ against $S_{fisher}$ we define $p(q)$ for all $q>0$ as

$$p(q)=p(q,S_{fisher},S_{cook}(q))=\liminf_{n\to\infty}\frac{1}{n}\sum_{i=1}^n a_i.$$

\begin{lemma}
\label{lemma_growth}Under the above assumptions, for each $q,x\geq 0$ we have $$Cook(S_{fisher},S_{cook}(q+x),q+x)\geq Cook(S_{fisher},S_{cook}(q),q)+x\cdot p(q).$$
\end{lemma}
\begin{proof}If the cook's type is $q+x$ and he plays the strategy $S_{cook}(q)$, his average surplus is at least $$Cook(S_{fisher},S_{cook}(q),q)+x\cdot p(q)$$ which must be at most $Cook(S_{fisher},S_{cook}(q+x),q+x)$.
\end{proof}

Morally, this estimate tells us that the derivative of $R(q)  = Cook(S_{fisher},S_{cook}(q),q)$ (if it exists) is at least $p(q)$. Also, this is the old good Spence-Mirrlees property.

It is not difficult to show that if $p(q)$ locally decreases with $q$ then is not possible that both fisher and cook are better when $q$ increases. Therefore it is natural to assume that $p(q)$ grows with $q$. Additionally, one is likely to expect that the higher the type $q$ of the buyer, the more frequently the deals are made. Also it is natural to assume that $\lim_{q\to\infty} p(q)=1$.

\subsection{Distortion at the top} It directly follows from Lemma~\ref{lemma_grows} that the following property holds. 

\begin{corollary}If $p(q)$ is non-decreasing function on $q$ with $\lim_{q\to\infty}p(q)=1$ then $$\lim_{q\to\infty}\frac{Fisher(S_{fisher},S_{cook}(q),q)}{Cook(S_{fisher},S_{cook}(q),q)}=0,$$
i.e the fraction that the fisher gets out of the total welfare $q$ tends to zero.
\end{corollary}

%Similar statement can be found in \cite{amin2013learning}. 
Note that using an arbitrary strategy $S_{fisher}$ of the fisher and the responses $S_{cook}(q)$ of the cook we constructed a function $p(q)$, the observable which measures the proportion of accepted deals. Now we switch our attention to this monotone function $p(q):[0,\infty)\to[0,1]$ and show how to use it in order to construct an incentive-compatible strategy $S_{fisher}$ of the fisher, such that the best response strategy $S_{cook}(q)$ of the cook will amount to exactly this given proportion $p(q)$ of accepted deals.

% Then,  Lemma~\ref{lemma_growth} shows that the strategy in Section~\ref{sec_mech} built by the function $p$ brings to the fisher a revenue asymptotically not worse than $S$.

\section{A new mechanism}
\label{sec_mech}
Our mechanism is introduced after the following informal motivation. 
\subsection{Commitment}

Let $S_{cook}^{naive}(q)$ be the following {\it naive strategy} for the cook: accept the price $Q_i$ if $Q_i\leq q$, refuse otherwise. For the cook of type $q$, playing $S_{cook}^{naive}(q)$ is not always optimal: the fisher quickly determines $q$ and the cook's surplus $q- Q_i$ will tend to zero.

On the other hand, imagine that the cook makes a {\it commitment}: he tells to the fisher that he is going to play $S_{cook}^{naive}(1)$ regardless the actions of the fisher. If the cook sticks to this strategy (fortunately for the fisher in many cases the cook has no such a commitment power), then the best response for the fisher is to always propose $Q_i=1$, other strategies bring strictly less revenue.  In any case, it is in the interest of the fisher to incentivise the cook not to play this strategy if the cook's type is much bigger than $q=1$. The only way to do this is to assure the cook that the higher prices he accepts, higher will be his average surplus.

If the fisher has a strategy, it makes sense to reveal it. Indeed, otherwise a strategic cook tricks the fisher and is afraid to buy even for a small $Q_i$ because this gives the fisher an information about the cook's type, and the fisher will use this information in an unknown way. It is proven that the absence of a fisher's commitement may crash the market \cite{devanur2015perfect,hart1988contract,schmidt1993commitment,immorlica2017repeated}. 

\subsection{Strategy background ideas, fisher's guess of the cook's type}
A strategy of the fisher (choosing the price $Q_{n+1}$ after the round $n$) can depend on the history of proposed prices $Q_i$ and the responses $a_i$ of the cook in all rounds $i$ preceding a given round. It seems reasonable for the fisher to squash this bunch of information to a number:  we suppose that before playing the round $i$ the fisher has a guess $q_i$ of the cook's type, and his next proposal depends only on $q_i$.

A good strategy $S_{fisher}$ of the fisher has the following qualities:
\begin{itemize}
\item {\it type adaptation}, $S_{fisher}$ should depend on $q_i$ and adapt to it (one may imagine that the true type $q$ of the cook slowly changes over time an it would be good if the fisher's strategy adapts to this change);
\item {\it rewarding}, to incentivize the cook, the surplus of the cook should be monotone with respect to $q_i$, and even more: since the cook can pretend that his type is lower than it is, the growth of the cook's surplus when the fisher's estimate $q_i$ increases must include what the cook can obtain by pretending that he is of a lower type;
\item {\it type confirmation}, the fisher should not incentivise the cook to pretend that his type is higher than it is, because stimulating to lie is not good by itself and because this again can crash the market and impose additional revenue loss for the fisher.
\end{itemize}

To disentangle all these features we take each of them to the extreme and propose each of them with some probability (because mixed strategies frequently work better in the context of learning or truthful mechanism design, cf. \cite{dobzinski2009power}). Start with {\it rewarding}: the fisher will sometimes propose $Q_i=0$, i.e. will give the object to the cook for free. To implement {\it type adaptation} it is enough to propose a price higher than the current estimate $q_i$, for example, $Q_{i}$ is drawn from a uniform distribution on $[q_i,1+q_i]$ (or any other distribution with the support on $[q,+\infty)$). To ensure that the cook is indeed of the type $q_i$, one might propose the price $q_i$ sometimes, and if the cook refuses, then set  $q_{i+1}<q_i$, which should cause a revenue loss for the cook.

\subsection{The mechanism, formally}
The fisher knows that $q\in[0,\infty)$ (the mechanism can be easily adapted to the case when the fisher knows for sure that $q$ belongs to a fixed interval). The fisher fixes any increasing function $p:[0,\infty)\to[0,1], p(0)=0, \lim_{q\to\infty}p(q)=1$.
Let $R$ be the solution of the following differential equation: $R(0)=0, R'(q)=p(q)$. The fisher commits to use the following strategy, all the ingredients ($p$ and the following algorithm) of which are announced to the cook in advance. We will show that the optimal response $S_{cook}(q)$ of the cook of type $q$ gives him $Cook(S_{fisher},S_{cook}(q),q)=R(q)$.

In each round the fisher has an estimate $q_n$ of $q$ ($q_0=0$). The fisher plays the following mixed strategy $S_{fisher}(q_n)$ randomly choosing between three types of prices in this round: 
\begin{itemize}
\item (with probability $1-p(q_n)$) the fisher plays {\it type adaptation}: the price $Q_n$ is uniformly drawn from $[q_n,q_n+1]$, accepting this plice the cook achieves that $q_{n+1}>q_n$. If $Q_n$ is accepted, then the fisher sets $q_{n+1}=Q_n$. If $Q_n$ is refused, then $q_{n+1}=q_n$.

\item (with probability $\frac{R(q_n)}{q_n}$) the fisher plays {\it rewarding}: the price $Q_n=0$, which is profitable for the cook, and the probability of this price appearing grows with $q_n$ thereby stimulating the cook to reveal his type.

\item (with probability $p(q_n)-\frac{R(q_n)}{q_n}$) the fisher plays {\it type confirmation}: the price $Q_n=q_n$, which is assumed to give no surplus to the cook, but refusing it incurs lowering the estimate of the type of the cook. If $Q_n$ is accepted, then $q_{n+1}=q_n$. If  $Q_n$ is refused, then we should set $q_{n+1}<q_{n}$. This can be done in many ways, e.g. let $k=\sum_{i=1}^n a_i$. Let us reorder the set $\{Q_i\}_{i=1}^n$, $$Q_1'\geq Q_2'\geq Q_3'\geq\dots Q_n'.$$ The fisher sets $q_{n+1}=\frac{Q_{n-k}'+Q_{n-k-1}'}{2}$. The intuition for this formula is as follows: it is natural to define $q_n$ (the empirical guess of $q$) as the average between the maximal accepted price and the minimal rejected price. But, depending on the strategy of the cook, the latter can be less than the former. So we reorder all proposed prices and choose $q_n$ such that the proportion of proposed prices higher than $q_{n+1}$ is equal to the proportion of refused prices.

\end{itemize} 

%
%by setting $q_i=\frac{Q+\max(Q,Q')}{2}$, where $Q$ is the maximal accepted price and $Q'$ is the minimal rejected price during the last ten days (not including today).

While the fisher plays  $S_{fisher}(q_n)$, the average surplus of the cook of type $q\ne q_n$, playing $S_{cook}^{naive}(q)$ is
$$(q-q_n)\left(p(q_n)-\frac{R(q_n)}{q_n}\right)+ q\left(\frac{R(q_n)}{q_n}\right)=(q-q_n)p(q_n)+R(q_n)<R(q),$$
where the last inequality follows from the definition of $R(q)$ if $q_n<q$, and from the convexity of $R(q)$ when $q_n>q$. This kind of inequality is common in Revenue Equivalence type theorems.

From this  we derive the following theorem.
\begin{theorem} 
For any strategy $S_{cook}(q)$ of the cook we have $$\mathbb E(Cook(S_{fisher},S_{cook}^{naive}(q),q))\geq \mathbb E(Cook(S_{fisher},S_{cook}(q),q)),$$
where $\mathbb E(\cdot)$ stands for the expectation (recall that the strategy $S_{fisher}$ is not pure).
\end{theorem}

\begin{proof}The cook can manipulate $q_n$. The incentives are as follows: 1) If $q_n<q$ then for the cook it is profitable because the cook gains additionally on accepting {\it type adaptation} price; 2) If $q_n>q$, then for the cook it is profitable because the {\it reward} price $Q_i=0$ appears more frequently.

By the mechanism, to sustain $q_n<q$ the cook has to always reject prices $Q_i>q_n$, i.e. this strategy can not be more profitable than $S_{cook}^{naive}(q_n)$. In order to make $q_n$ bigger than $q$, the cook should accept {\it type adaptation} prices (which are rare, so the process will take a long time) and {\it type confirmation} prices, both have negative surplus for the cook, so in the expectation it is not profitable. Then the cook may enjoy reward prices for some time, but refusing a status confirmation price results in setting $q_n$ lower. 

To summarise, due to the strategy of the fisher, in order to make the fisher believe that $q$ is $q'$ the cook should, in fact, play a strategy $S_{cook}^{naive}(q')$ for a substantial time, which is not profitable on average by construction.
\end{proof}

Surely, a risk loving cook can play $S_{cook}^{naive}(x)$ with $x>q$ and this can be more profitable than playing $S_{cook}^{naive}(q)$ on short sequences of rounds. %Therefore it would be nice to find estimates on the probability of $B(x)$ being profitable on a a sequence of rounds of given length. 

% In the previous section, under the mildest assumptions we obtained estimates on the average surplus of the cook. Let us construct a mechanism (taking $p$ as an input) which can not be improved, according to these estimates.

\section{Discussion}
The main idea of this paper is to put the situation upside down: for a third party observing the game the most salient number which can be extracted from the game is the proportion $p$ of accepted deals. It is a function $p(q)$ on the type $q$ of the cook. Then to construct the mechanism we use this function to produce the strategy $S_{fisher}=S_{fisher}(p(\cdot))$ of the fisher, and if the cook uses his best response to $S_{fisher}$, then the proportion of accepted deals is exactly $p(q)$. Any good strategy must have the following features: it should reward the cook for the revelation of his true type (and here we see a common phenomenon: higher the type, bigger should be the reward), it should adapt to the cook's type, and it should not incentivise the cook to pretend that his type is higher than it is, because in this case the fisher should be giving the cook a bigger reward. Our mechanism satisfies these three requirements.

%We provide an estimate for the fisher's average revenue as a function on $p(q)$: when the current fisher's estimate of the cook's type $q$ is $q_n$, only with probability $p(q_n)$ the proposed deal will be accepted. If $p(q)$ is always zero, the cook would not accept the game, if $p(q)$ is always one, the cook will lie and pay almost zero. Therefore $p(q)$ seems to be a {\bf reasonable (if not the unique) control parameter} which gives right incentives to the cook. Then a mechanism which asymptotically achieves this estimate is proposed, again, as a function on $p(q)$.

As far as we know %\footnote{Please, send me any feedback and suggestions to nikaanspb on gmail on com} 
the proposed mechanism is novel (and very simple). Another distribution for mixed strategies may be considered in the fisher's strategy, but they all should have an atom at $q$: indeed, if the surplus is growing with $q$, the cook has an incentive to pretend that his type is higher than it is. To prevent it we must always check (by offering the price $q$) that the cook is indeed of type $q$ or lower. In the {\it contract theory} paradigm (cf. recent \cite{malcomson2016relational}) one can give the following metaphor: 1) an employee should have opportunities to show that she is more capable that her status suggests, 2) higher the status, more free benefits (ratchet effect), 3) there should be always work to do, to confirm her high status.

The idea that in order to reveal the true type of the cook, the fisher must give him the substantial part of the total welfare (almost $100\%$ of welfare when the cook's type is huge) is not new and can be traced back to \cite{myerson1981optimal}.

Note another advantage of the proposed mechanism: if the valuation $q$ of the cook changes, the algorithm can be easily adapted -- for example, we may set the rule for $q_n$ taking into account only the last $100$ deals. Similar statements can be found in \cite{amin2013learning} (Section 6), but under a different abstract disguise, in another context, and without a concrete mechanism. The mechanism is credible: by playing enough time in a stable position (meaning that $q_i$ does not change) the cook can calculate the frequency of reward, adaptation, and confirmation prices, and compare them with the probabilities which can be derived from the function $p(\cdot)$ which was announced in advance. 

\subsection{How the fisher chooses $p(q)$?} Indeed, instead of saying that the fisher knows the distribution of $q$ we need to choose $p(q)$, which, again, means to make a guess about distribution. A subtle difference is that it is hard to guess the distribution if we suspect that (surely, with small probability, but what is the order of magnitude?) $q$ can be really huge. Also, our mechanism can be used to find $q$ of a given customer, and the customers are reluctant to reveal that information if they know that the game is playing for many rounds in the future.

\subsection{Further questions} How do real people behave playing this game? (cf. \cite{ostrovsky2011reserve}) Let us say, two participants play 20 rounds, and it is known that $q$ belongs to a given interval, but the distribution of $q$ is not known to the fisher role player. The players should be paid: e.g. a fixed amount of money is equal to $20q$, and then the players get the corresponding proportion of that money according to their results. In such a way both are interested in making the deal (any rejection is a waste of $q$ money that they should ``divide'' among them). I guess that the fairness will not be an issue here because actual $q$ may be small or may be big, the fisher-player has no estimate of it and subsequently can not perceive a deal as fair or not until the end of game when $q$ is revealed.
 
A multi-person game may be considered in the spirit of \cite{athey2013efficient}: let the fisher face several cooks (of a priori different types), the game be infinite, and there be no discounting. It seems that for any mechanism there will be a distortion at the top estimate, but only for the cook of the highest type, whose presence coerces all other cooks to reveal their types. Finally, it would be interesting to find a formal setup with only one buyer where there is a force similar to the presence of a higher type buyer but weaker than that (an instance of such a force is a common trick when a seller tells you that there is another buyer who agreed to pay way more than you, but tomorrow).

\bibliography{../../Dropbox/bibliography.bib}

\begin{thebibliography}{10}

\bibitem{akbarpour2018credible}
M.~Akbarpour and S.~Li.
\newblock Credible mechanisms.
\newblock In {\em Proceedings of the 2018 ACM Conference on Economics and
  Computation}, pages 371--371. ACM, 2018.

\bibitem{amin2013learning}
K.~Amin, A.~Rostamizadeh, and U.~Syed.
\newblock Learning prices for repeated auctions with strategic buyers.
\newblock In {\em Advances in Neural Information Processing Systems}, pages
  1169--1177, 2013.

\bibitem{athey2013efficient}
S.~Athey and I.~Segal.
\newblock An efficient dynamic mechanism.
\newblock {\em Econometrica}, 81(6):2463--2485, 2013.

\bibitem{babaioff2011posting}
M.~Babaioff, L.~Blumrosen, S.~Dughmi, and Y.~Singer.
\newblock Posting prices with unknown distributions.
\newblock In {\em In ICS}. Citeseer, 2011.

\bibitem{babaioff2010detail}
M.~Babaioff, S.~Dughmi, and A.~Slivkins.
\newblock Detail-free, posted-price mechanisms for limited supply online
  auctions.
\newblock 2010.

\bibitem{bergemann2005robust}
D.~Bergemann and S.~Morris.
\newblock Robust mechanism design.
\newblock {\em Econometrica}, 73(6):1771--1813, 2005.

\bibitem{bergemann2008pricing}
D.~Bergemann and K.~H. Schlag.
\newblock Pricing without priors.
\newblock {\em Journal of the European Economic Association}, 6(2-3):560--569,
  2008.

\bibitem{devanur2015perfect}
N.~R. Devanur, Y.~Peres, and B.~Sivan.
\newblock Perfect bayesian equilibria in repeated sales.
\newblock In {\em Proceedings of the twenty-sixth annual ACM-SIAM symposium on
  Discrete algorithms}, pages 983--1002. Society for Industrial and Applied
  Mathematics, 2015.

\bibitem{dobzinski2009power}
S.~Dobzinski and S.~Dughmi.
\newblock On the power of randomization in algorithmic mechanism design.
\newblock In {\em Foundations of Computer Science, 2009. FOCS'09. 50th Annual
  IEEE Symposium on}, pages 505--514. IEEE, 2009.

\bibitem{drutsa2017horizon}
A.~Drutsa.
\newblock Horizon-independent optimal pricing in repeated auctions with
  truthful and strategic buyers.
\newblock In {\em Proceedings of the 26th International Conference on World
  Wide Web}, pages 33--42. International World Wide Web Conferences Steering
  Committee, 2017.

\bibitem{drutsa2017consistency}
A.~Drutsa.
\newblock On consistency of optimal pricing algorithms in repeated posted-price
  auctions with strategic buyer.
\newblock {\em arXiv preprint arXiv:1707.05101}, 2017.

\bibitem{goldner2016prior}
K.~Goldner and A.~R. Karlin.
\newblock A prior-independent revenue-maximizing auction for multiple additive
  bidders.
\newblock In {\em International Conference on Web and Internet Economics},
  pages 160--173. Springer, 2016.

\bibitem{hart1988contract}
O.~D. Hart and J.~Tirole.
\newblock Contract renegotiation and coasian dynamics.
\newblock {\em The Review of Economic Studies}, 55(4):509--540, 1988.

\bibitem{immorlica2017repeated}
N.~Immorlica, B.~Lucier, E.~Pountourakis, and S.~Taggart.
\newblock Repeated sales with multiple strategic buyers.
\newblock In {\em Proceedings of the 2017 ACM Conference on Economics and
  Computation}, pages 167--168. ACM, 2017.

\bibitem{kleinberg2003value}
R.~Kleinberg and T.~Leighton.
\newblock The value of knowing a demand curve: Bounds on regret for online
  posted-price auctions.
\newblock In {\em null}, page 594. IEEE, 2003.

\bibitem{malcomson2016relational}
J.~M. Malcomson.
\newblock Relational incentive contracts with persistent private information.
\newblock {\em Econometrica}, 84(1):317--346, 2016.

\bibitem{mohri2014optimal}
M.~Mohri and A.~Munoz.
\newblock Optimal regret minimization in posted-price auctions with strategic
  buyers.
\newblock In {\em Advances in Neural Information Processing Systems}, pages
  1871--1879, 2014.

\bibitem{mohri2015revenue}
M.~Mohri and A.~Munoz.
\newblock Revenue optimization against strategic buyers.
\newblock In {\em Advances in Neural Information Processing Systems}, pages
  2530--2538, 2015.

\bibitem{myerson1981optimal}
R.~B. Myerson.
\newblock Optimal auction design.
\newblock {\em Mathematics of operations research}, 6(1):58--73, 1981.

\bibitem{ostrovsky2011reserve}
M.~Ostrovsky and M.~Schwarz.
\newblock Reserve prices in internet advertising auctions: A field experiment.
\newblock In {\em Proceedings of the 12th ACM conference on Electronic
  commerce}, pages 59--60. ACM, 2011.

\bibitem{schmidt1993commitment}
K.~M. Schmidt et~al.
\newblock Commitment through incomplete information in a simple repeated
  bargaining game.
\newblock {\em Journal of Economic Theory}, 60:114--114, 1993.

\bibitem{tang2018price}
P.~Tang and Y.~Zeng.
\newblock The price of prior dependence in auctions.
\newblock In {\em Proceedings of the 2018 ACM Conference on Economics and
  Computation}, pages 485--502. ACM, 2018.

\end{thebibliography}
\bibliographystyle{abbrv}

\end{document}